\documentclass[aps,pra,amsmath,amssymb,reprint,nofootinbib]{revtex4-2}

\usepackage{amsmath}

\usepackage{graphicx}
\usepackage{dcolumn}
\usepackage{bm}

\usepackage[utf8]{inputenc}
\usepackage[T1]{fontenc}
\usepackage{mathptmx}
\usepackage{etoolbox}

\usepackage{amsmath}
\usepackage{amsthm}
\usepackage{tikz}
\usepackage{extpfeil}
\usepackage{tikz-cd}
\usetikzlibrary{snakes}
\usepackage{physics}
\usepackage{caption}
\usepackage{subcaption}

\usepackage{dsfont}

\usepackage{appendix}
\usepackage[breaklinks]{hyperref}
\hypersetup{
    colorlinks=true,
    linkcolor=blue,
    filecolor=magenta,
    urlcolor=cyan,
}

\newextarrow{\xbigtoto}{{20}{20}{20}{20}}
   {\bigRelbar\bigRelbar{\bigtwoarrowsleft\rightarrow\rightarrow}}
\newextarrow{\xbigto}{{20}{20}}
   {\bigRelbar{\bigtwoarrowsleft\rightarrow}}
   
\newtheorem{theorem}{Theorem}
\newtheorem{definition}[theorem]{Definition}

\newtheorem{example}{Example}

\usepackage{ragged2e}
\usepackage{url}

\usepackage{breakurl}

\begin{document}


\title[Differential Geometry of Contextuality]{Differential Geometry of Contextuality}

\author{Sidiney B. Montanhano}
\email[]{s226010@dac.unicamp.br}
\affiliation{Instituto de Matemática, Estatística e Computação Científica, Universidade Estadual de Campinas, 13083-859, Campinas, São Paulo, Brazil}

\date{\today}

\begin{abstract}
Contextuality has long been associated with topological properties. In this work, such a relationship is elevated to identification in the broader framework of generalized contextuality. We employ the usual identification of states, effects, and transformations as vectors in a vector space and encode them into a tangent space, rendering the noncontextual conditions the generic condition that discrete closed paths imply null phases in valuations, which are immediately extended to the continuous case. Contextual behavior admits two equivalent interpretations in this formalism. In the geometric or intrinsic-realistic view, termed ``Schrödinger'', flat space is imposed, leading to contextual behavior being expressed as non-trivial holonomy of probabilistic functions, analogous to the electromagnetic tensor. As a modification of the valuation function, we use the equivalent curvature to connect contextuality with interference, noncommutativity, and signed measures. In the topological or participatory-realistic view, termed ``Heisenberg'', valuation functions must satisfy classical measure axioms, resulting in contextual behavior needing to be expressed in topological defects in events, resulting in non-trivial monodromy. We utilize such defects to connect contextuality with non-embeddability and to construct a generalized Vorob'ev theorem, a result regarding the inevitability of noncontextuality. We identify in this formalism the contextual fraction for models with outcome-determinism, and a pathway to address disturbance in ontological models as non-trivial transition maps. We also discuss how the two views for encoding contextuality relate to interpretations of quantum theory.
\end{abstract}

\maketitle

\tableofcontents

\section{Introduction}
\label{1}

Contextuality, the ultimate form of non-classicality, has many diverse mathematical approaches. Each approach was built for a specific purpose or strategy to exploit its characteristics, and some were developed long before being identified as contextuality. They all start from the encoding of physical systems into some mathematical structure that cannot be represented by another structure called classical. There are already examples of topological representations \cite{Abramsky_2011,Abramsky_2012,abramsky_et_al:LIPIcs:2015:5416,okay2017topological,2020Okay,Sidiney_2021}, and algebraic representations \cite{BirkhoffvonNeumann,PhysicsPhysiqueFizika.1.195,Kochen1975,Gleason}. Some have evolved into a geometrical representation due to the relationship between inequalities and convex sets \cite{Cabello_2014,amaral2017geometrical}, and others seek a foundation in measure theory \cite{Dzhafarov2015ContextualitybyDefaultAB}. Other notions \cite{1994SORKIN,Spekkens2008,Schmid2021} are known to be related to standard contextuality, and they are more or less explored in the literature.

Non-classicality has an incredible number of applications, and more are being presented each day. Contextuality is the fuel for this revolution. It is known that contextuality is the origin of quantum behavior \cite{doring2020contextuality}, and it is the generalization of the famous notion of nonlocality \cite{Abramsky_2011}. It is necessary for any computational advantage over classical computers \cite{shahandeh2021quantum}, and it is explicitly the ``magic'' required for some types of quantum computers \cite{Howard2014}. But this phenomenon is not just a resource for technological applications. Understanding contextuality in a more general formulation is essential to understand why and how we live in a quantum reality, and whether we need to search for more general theories than quantum theory itself. This fundamental exploration ultimately aims to establish the framework where future theories and technologies will be built.

In this work, we will explore the geometric or topological origins of contextuality. We will use its more generic version, generalized contextuality \cite{Spekkens2005}, restricting as necessary to treat tools from other approaches. The general strategy will be to rethink the operational equivalences of the ontic representation as loops with discrete parts in the tangent space of a suitable manifold, usually piecewise linear given by the elements of a process (for example, an extension of the set of effects). The noncontextual condition becomes the preservation of probabilistic valuation maps for these loops, thus, without the presence of non-trivial phases.

A relevant ontic representation will be the one present in generalized probability theories \cite{Janotta2014}. Generally, generalized probability theories are constructed with a finite set of extremal effects, and demonstrations of contextual behaviors use them. But there are theories where the set of such effects is infinite and, like quantum theory, even continuous. This implies the possibility of different types of operational equivalences. The framework we build where all operational equivalences are explored on equal footing. The condition of contextual behavior, described through discrete differential geometry, ends up generalizing to the domain of differential geometry for the continuous case, again as non-trivial phases of probabilistic valuation maps.

With this framework, we present two ways to interpret contextuality, depending on the choice of how to encode the physical system. The violation of phase triviality in valuation results from holonomy or monodromy, respectively, linked to the intrinsic-realistic and participatory-realistic views of the theory. Holonomy follows from a geometric cause, by imposing that the ontological set be classically complete, which imposes on valuations the correction of contextual behavior, expressed by a curvature term. Monodromy follows from a topological cause, by imposing that the valuation must be classical in the sense of measure theory, expressed by topological defects in the ontological set. They are equivalent and have dual interpretations of the fact that we have lost classicality.

We will use these two equivalent interpretations to explore the relationship between contextual behavior and different notions of non-classicality. Using the geometric version, called ``Schrödinger'', we start by imposing outcome determinism, making the contextual fraction, a quantifier of contextuality \cite{Abramsky_2017}, to be defined in its usual form. We show that the resource it is quantifying is nothing more than the curvature term that corrects valuation. The same occurs with the interference term in the sense of quantum measure theory \cite{1994SORKIN,sorkin1995quantum}, where the interference part that cannot be explained classically also arises from the curvature term. The need for signed measures for valuation maps \cite{Spekkens2008,Abramsky_2011} follows analogously as we force the corrected valuation to be treated in the same way as a classical measure, even though we have already imposed the non-existence of topological defects. Using the topological version, called ``Heisenberg'', we identify it as the cause of the impossibility of incorporation into a classical mathematical structure \cite{Schmid2021,schmid2020structure}. The topological view also offers a generalization of the famous Vorob'ev theorem \cite{Vorobyev_1962}, which characterizes the inevitability of noncontextuality, while the geometric view establishes a relationship between transition maps and disturbance \cite{Sidiney_2021,Amaral_2019}. They also allow for a clearer understanding of how quantum interpretations deal with contextuality, in a way mapping the madness of interpretations \cite{Cabello_2017}.

This work is divided as follows. In section \ref{2}, there is a presentation of generalized probability theories (\ref{2.1}), the notion of ontic representation, and the valuation functions for each part of a usual process (\ref{2.2}), as well as the conditions of noncontextuality for each of these parts (\ref{2.3}). A brief presentation of differential geometry (\ref{3.1}) and discrete differential geometry (\ref{3.2}) is given in section \ref{3}. In section \ref{4}, contextuality is identified as a phase of the valuation functions. Beginning with the formal identification (\ref{4.1}), we construct its two interpretations. In the geometric view (\ref{4.2}), we have the notion of contextual curvature for a generalized probability theory. In the topological view (\ref{4.3}), we have the relationship between contextuality and non-trivial topology. As equivalent interpretations, they can be translated into each other (\ref{4.4}). The tour of already known non-classical phenomena and their relations with contextuality begins in section \ref{5}. We start with the identification of contextual fraction (\ref{5.1}), the interference in quantum measure theory (\ref{5.2}), and with the special example of quantum theory and its dependence on Planck's constant. We then explore the necessity of signed measures and the impossibility of embedding the process in a classical mathematical object (\ref{5.3}), and generalize the Vorob'ev theorem with the topological view (\ref{5.4}). The zeroth order process, given by disturbance, is explained in a generalized framework for generalized probability theories, adding a new term to the valuation decomposition (\ref{5.5}). We finalize this section with an exposition of the main interpretations of quantum theory and its relation to the interpretations of the origin of contextuality (\ref{5.6}). Section \ref{6} follows, where we discuss the results and potential future avenues.

\section{Generalized Contextuality and Generalized Probability Theories}
\label{2}

Let's review here some concepts regarding generalized probability theories. Such theories will be important for formalizing sets of operational objects that are accessible for probabilistic valuation. They will serve as a substrate for subsequent analyses.

The objective of operational probability theories is to provide an operational description of physical theories, serving as an initial construction for their purely operational depiction \cite{schmid2020structure}. We will work with standard probabilities and with only one system, so it is unnecessary to define how to combine systems together. This restriction categorizes operational probability theories as generalized probability theories, or GPTs.

\subsection{Generalized Probability Theories}
\label{2.1}

A GPT can be described as a category. A category is given by a class of objects and a class of morphisms between each pair of objects, with one as the source and the other as the target, where there exists composition of morphisms, and such composition is associative. In the case of GPTs, the morphisms define operations that represent physical operations between the objects, and we have a trivial object that serves as a fundamental element for the construction of processes and their valuation. For a presentation of category theory aimed at applications in quantum theory, see Refs. \cite{coecke2008introducing, Coecke2011}.

The operations from the trivial object, which we denote by $\bot$, to any other object are called states, denoted generally by $P$ or categorically by $f\leftarrow \bot$, and we will denote them as $\ket{P}$ for reasons that will become clear later on. The set of operations from objects other than the trivial one to it are called effects, denoted generally by $E$ or categorically by $\bot\leftarrow f$, and we will denote them as $\bra{E}$.

The other morphisms between two non-trivial objects can be understood as representations of transformations. Any transformation $f\leftarrow g$ can be thought of as a transformation that takes the state $f\leftarrow\bot$ to the state $g\leftarrow\bot$, which we can denote as $T\ket{P_f}=\ket{P_g}$, with $T$ denoting the transformation in question. The same reasoning applies analogously to effects. This identifies physically interesting transformations as functions of the set of states to itself, or equivalently as functions of the set of effects to itself.

The automorphisms of the trivial object possess a structure of scalars, usually taken as a semiring or semifield, such as the Boolean semiring $\mathbb{B}$ given by ${0,1}$, and the probabilistic semifield $\mathbb{R}^+$ given by $[0,1]$. Here, we will focus on the probabilistic semifield $\mathbb{R}^+$. The sets of states $\mathcal{P}$, transformations $\mathcal{T}$, and effects $\mathcal{E}$ provide us with the probabilities of a process in the system they represent through a function:
\begin{equation}
p:\mathcal{P}\times\mathcal{T}\times\mathcal{E}\to[0,1]::(P,T,E)\mapsto p(E|T,P)
\end{equation}
interpreted as the probability of obtaining the outcome $E$ when starting with a state $P$ that underwent a transformation $T$. The composition of operations used in the argument of the function $p(E|T,P)$ can be understood as a path from the trivial object to itself, passing through the operations that define $P$ first, then the operations that define $T$, and finally the operation of $E$, closing the loop in the category of operations. This identifies $p$ as a function of loops passing through $\bot$ to the set of scalars, in our case, the probabilistic one.

Usually, we can use the bracket notation that we have already introduced to write
\begin{equation}
p(E|T,P)=\bra{E}T\ket{P},
\end{equation}
indicating that this is an identification with a loop of processes and in analogy with quantum theory and linear algebra. To turn this analogy into an identification, one imposes that the function $p$ preserves mixtures of operations, i.e., the convex combination with scalar coefficients. From the preservation of mixtures, one can extend it to obtain linearity of $p$. With linearity, we can represent states and effects in a vector space, and transformations as linear maps acting on them such that they preserve the sets of states and effects. The identification of the bracket notation occurs with the identification of $p$ with an inner product in this vector space. 

\begin{example}
    Let's explicitly illustrate, as an example, the GPT structure of a qubit, the two-dimensional quantum system. Naturally, we can begin with the Bloch sphere representation, where a state $\rho$ is encoded as a vector $\vec{a}$ on a unit sphere in $\mathbb{R}^3$, thus $|\vec{a}|\leq 1$. Here, we have
\begin{equation}
\rho=\frac{1}{2}(I+\vec{a}\cdot\vec{\sigma}),
\end{equation}
where $I$ is the identity matrix and $\vec{\sigma}$ is the vector given by the Pauli matrices basis. The origin of the space can be identified with $\bot$, and every vector from the origin to a point within the sphere determines the process defining the state given by that point.

    To define the effects, let us use the fact that quantum theory satisfies the condition of strong duality \cite{PhysRevLett.108.130401}. This property allows the set of states and effects to be represented in the same way, with the difference that the process occurs in the opposite direction, from the point on the sphere to the origin, but also uniquely defining a vector to represent the effect.

    Another property of quantum theory is that its set of physical transformations is maximal, meaning that all possible transformations are physically feasible. Due to linearity, we have that a transformation will be represented by a linear map $\overline{T}$ with induced norm $|\overline{T}|\leq 1$.

    With all this representation in a space $\mathbb{R}^3$, it is natural that probabilistic valuation be given by an inner product, which indeed occurs by the Born rule. We have that the trace can be calculated as
\begin{equation}
p(E|T,P)=\Tr{eT\rho}=\frac{1}{4}(1+\vec{a}e^{\intercal}\overline{T}\vec{a}{\rho}).
\end{equation}
\end{example}

There are further intriguing mathematical intricacies regarding this construction. For a more formal exposition, see Refs. \cite{Amaral:2014mlu, Janotta2014, Muller2021, Selby2021}. For a detailed treatment of linearity, particularly consult Ref. \cite{Amaral:2014mlu}.

\subsection{Ontic Representation}
\label{2.2}


A (classical) ontic representation of a (physical) theory involves embedding its vector space representation into classical probability theory. In such a theory, a simplex is defined as a set of states, and its dual as the set of effects. The primary aim of an ontic representation is to expand upon the original theory by refining the involved variables, thus enabling an explanation of the statistics of the GPT as a sub-model of classical theory.

\begin{example}
As a second example of GPT, which is important for addressing ontic representation, we have a theory of classical probabilities. In it, the states already directly represent the probabilities. In a classical system generated by $n$ possible outcomes, we have a vector in $\mathbb{R}^n$ whose elements are nonnegative real numbers that sum to 1 as a state. In other words, the set of states $\mathcal{P}$ is given by a simplex in $\mathbb{R}^n$.

Probabilistic valuation is nothing more than the entries of the vector, and by taking the maximal set of effects, we have that the set of effects $\mathcal{E}$ will be given by the dual set of vectors. This defines a hypercube in $\mathbb{R}^n$, with vertices on each of the coordinate axes, at the origin, and on the vector with entries being $1$, which acts as the identity vector.

Transformations are also taken as the maximal set, given by linear maps that preserve probabilities. These are stochastic maps, or Markov maps.
\end{example}

We can always do an embedding into a classical GPT, and the probability of the model will be given by the chain rule
\begin{equation}
p(E_{r}|T_{t},P_{s})=\sum_{\lambda,\lambda'}\xi(E_{r}|\lambda')\Gamma(\lambda',T,\lambda)\mu(\lambda|P_{s}),
\end{equation}
the valuation functions $\xi$, $\Gamma$, $\mu$, and the set of ontic variables denoted by $\Lambda$. However, this does not guarantee classicality, as one needs to impose conditions on the valuation functions to ensure they do not violate any classical behavior.

We need to impose the independence of the measurements to which the effects belong. But note that such representation is independent of measurements, once the function $\xi$ has the outcomes as its domain, in the form of a set of effects. Thus, it is restricted to non-disturbing models\footnote{No-disturbance is defined in the intersection of contexts: if the valuation of an intersection (when one defines it) $\xi(C\cap C')$ is independent of $C$ and $C'$ for all pairs of contexts, then $\mathcal{M}$ is said to be non-disturbing. It is related to parameter-independence \cite{Brandenburger_2008,barbosa2019continuousvariable}.}, or in other words, for measurements $m$ and $n$, it holds
\begin{equation}
\xi(E_{r}|\lambda,m)=\xi(E_{r}|\lambda,n)=\xi(E_{r}|\lambda),
\end{equation}
fixing the conditions for embedding to preserve classical GPT.

At this level of generality, outcome-determinism\footnote{Outcome-determinism is defined in the valuation: the outcomes defined on the contexts where the distribution is defined can be explained in a deterministic way. It is equivalent to restricting our interest to ideal measurements, which one can easily criticize when thinking in empirical applications \cite{Spekkens2014}.} is not required, unless one wants to use factorizability as a condition for noncontextuality, as in the Sheaf Approach \cite{Wester_2018}\footnote{Outcome-determinism implies that $\xi:\mathcal{E}\times\Lambda\to\{0,1\}$, thus codifying the determinism of this valuation.}. Once the outcomes are fundamental in this framework, represented in the events, and deal with the non-classicality of all the steps in a physical process, it is natural that this is the strongest framework to construct new generalized models up to the limitation by no-disturbance.

\subsection{Generalized Contextuality}
\label{2.3}

Generalized contextuality \cite{Spekkens2005} deals with preparations, transformations, and unsharp measurements (the latter equivalent to effect algebras defined by the set of effects in a GPT). It is a method to investigate the classicality of a system through operational equivalences. This approach is based on Leibniz's principle of indiscernibles, which states that if two objects possess the same properties, and therefore are indiscernible, then they are identical.

Generalized contextuality \cite{Spekkens2005} deals with preparations, transformations, and unsharp measurements (the latter equivalent to effect algebras defined by the set of effects in a GPT). It is a method to investigate the classicality of a system through operational equivalences. This approach is based on Leibniz's principle of indiscernibles, which states that if two objects possess the same properties, and therefore are indiscernible, then they are identical.

Using linearity to represent the construction of the final process from basic processes by their summation with certain coefficients, and rearranging the equation to make the linear dependence explicit, we find that for the states, the effects, and the transformations, operational equivalences can be expressed as linear conditions \cite{selby2021contextuality}
\begin{equation}
\label{basestates}
    \sum_{s}a_{s}^{(\alpha)}P_{s}=0,
\end{equation}
\begin{equation}
\label{baseeffects}
    \sum_{r}b_{r}^{(\beta)}E_{r}=0,
\end{equation}
\begin{equation}
\label{basetrasformations}
    \sum_{t}c_{t}^{(\tau)}T_{t}=0,
\end{equation}
indexed by $\alpha$, $\beta$, and $\tau$. This form will be important for what follows in the next sections. Note that although the initial construction imposes that the processes remain within the set of processes, and therefore limits the values of the coefficients, the reorganization in the linear conditions above allows the coefficients to turn the vectors into objects outside the set of processes.

\begin{definition}
    A theory is noncontextual if for a classical ontic representation the operational equivalences are preserved in the probabilities given by the valuation maps
\begin{equation}
\label{NCstates}
\sum_{s}a_{s}^{(\alpha)}\mu(\lambda|P_{s})=0,
\end{equation}
\begin{equation}
\label{NCeffects}
\sum_{r}b_{r}^{(\beta)}\xi(E_{r}|\lambda')=0,
\end{equation}
\begin{equation}
\label{NCtransformations}
\sum_{t}c_{t}^{(\tau)}\Gamma(\lambda',T_{t},\lambda)=0,
\end{equation}
for all $\lambda$ and $\lambda'$\footnote{This approach is more refined than other frameworks because of its operational interpretation, which allows the exploration of non-classicality beyond measurement and effects. The outcome-determinism is not imposed, but when imposed this approach is equivalent to the Sheaf Approach to contextuality \cite{Staton2015}.}.

\end{definition}

As the ontic representation is an embedding in a classical GPT, the ontic space $\Lambda$ is a simplicial set. The conditions of noncontextuality stated above assert that the original theory, its states, effects, and transformations, can be embedded in $\Lambda$, and its probabilities encoded in it as a coarse graining without violating classical probability theory in the sense of the Kolmogorov axioms \cite{schmid2020structure, Schmid2021}. Interestingly, this is equivalent to there being no need for negative values for the functions $\xi$, $\Gamma$, and $\mu$ when represented in an embedding as described above \cite{Spekkens2008}.

A property of valuation functions is their linearity within the convex set of objects in the domain. As an example, let's confine ourselves to the set of effects $\mathcal{E}$, though the same argument holds for states and transformations. Let $A,B\in\mathcal{E}$, and $A+B\in\mathcal{E}$, then
\begin{equation}
    p(A+B)=p(A)+p(B)
\end{equation}
if and only if
\begin{equation}
    \xi(A+B|\lambda')=\xi(A|\lambda')+\xi(B|\lambda'),
\end{equation}
which follows from the definition of ontic representation. Another way to see this is to note that if $A+B\in\mathcal{E}$, then $\{A,B,\mathds{1}-(A+B)\}$ is a valid measurement, thus
\begin{equation}
\begin{split}
    1=&\xi(\mathds{1}|\lambda')\\
    =&\xi(\mathds{1}-(A+B)|\lambda')+\xi((A+B)|\lambda')\\
    =&\xi(\mathds{1}-(A+B)|\lambda')+\xi(A|\lambda')+\xi(B|\lambda')
\end{split}
\end{equation}
and linearity follows, since $\xi(C|\lambda')+\xi(\mathds{1}-C|\lambda')=1$ for every effect $C$. But there is no guarantee that such linearity holds outside $\mathcal{E}$, as the following example shows.

\begin{example}
The Wigner's representation of quantum mechanics is an ontic representation of quantum theory, and it is linear for mixed states. However, as explained in Ref. \cite{doi:10.1119/1.2957889}, $W_{\psi}=W_{\alpha}+W_{\beta}$ with $\psi=\psi_{\alpha}+\psi_{\beta}$ generally does not hold. On the other hand, this holds for classical theories, which follow from the Kolmogorov axioms for probabilities.
\end{example}

\section{Differential Geometry and Discrete Differential Geometry}
\label{3}

In the previous section, we introduced the concepts of GPT, ontic representation, and how noncontextuality is expressed in such a formalism. In this section, we will present tools to glimpse what contextuality is actually doing in a physical model.

\subsection{Differential Geometry}
\label{3.1}

Let's quickly introduce the main concepts of differential geometry that we will use in the upcoming sections. The connection with what we have already presented will occur in the second part of this section, where we discretize what we will present here. With this, we will be able to present a characterization of contextuality in differential geometry in Section \ref{4} that can be immediately generalized to the continuous case. This will have important implications for the applications presented in Section \ref{5}.

We begin with smooth manifolds, which are locally similar to a vector space so that we can utilize calculus. We define a smooth manifold $\mathcal{M}$ as a topological space with topology $\tau$ that is covered by domains of homeomorphisms called charts $\varphi:U\in\tau\to\mathbb{R}^n$, such that for every pair of charts $\varphi_a$, $\varphi_b$, the transition map $\varphi_a\circ\varphi_b^{-1}$ is smooth. A set of charts that cover $\mathcal{M}$ is called an atlas, and it fully describes $\mathcal{M}$.

With this structure, we can work with the tangent vector space at a point $p\in\mathcal{M}$, denoted by $T_p\mathcal{M}$, which is given by a chart containing it that serves as a neighborhood. As in $\mathbb{R}^n$, we can define directional derivative $\partial_{k}=\frac{\partial}{\partial x_k}$. We are interested in studying how infinitesimal objects behave in this local environment and what properties they exhibit when seeking to extend them globally throughout $\mathcal{M}$. Since each point in $\mathcal{M}$ has its tangent vector space, we have a tangent bundle $T\mathcal{M}$.

The dual vector space of a tangent space is called the cotangent space, denoted by $T_p^*\mathcal{M}$, which is also defined at each point and generates a bundle $T^*\mathcal{M}$ called the cotangent bundle. It is important to note that this duality depends on which value we want to obtain from the application of an object from the cotangent space to the tangent space. Here we will deal with values in $\mathbb{R}$, meaning that if $\bra{\pi}\in T_p^*\mathcal{M}$ and $\ket{\omega}\in T_p\mathcal{M}$, then $\bra{\pi}\ket{\omega}\in\mathbb{R}$. The objects of $T_p^*\mathcal{M}$ are the covectors, and we can also define infinitesimal elements $dx_k$ dual to differentials $\partial_{k}$, known as differential $1$-forms.

The next step will be to think of $\partial$ and $d$ as operators by themselves. Continuing with the informality level of this presentation, $\partial$ acts as a boundary operator on a region, effectively reducing its dimension by one unit. Notice that each piece has its orientation, which is captured by the differential. We can define $\mathcal{C}_n(\mathcal{M})$ as the entire finite set of $n$-dimensional pieces of $\mathcal{M}$. The boundary operator is nothing more than a map $\partial: \mathcal{C}_n \to \mathcal{C}_{n-1}$, and since the boundary of a boundary is always an empty set, we have $\partial\partial = 0$, and we obtain the chain complex
\begin{equation}
    \begin{tikzcd}[row sep=tiny]
0 & \arrow{l}{\partial_0} \mathcal{C}_{0}(\mathcal{M}) & \arrow{l}{\partial_1} \mathcal{C}_{1}(\mathcal{M}) & \arrow{l}{\partial_2} \mathcal{C}_{2}(\mathcal{M}) & \arrow{l}{\partial_3} \dots
\end{tikzcd}
\end{equation}

We can explore the topology of $\mathcal{M}$ through its pieces and the natural group structure of these pieces via homology groups. First, we define the kernel of $\partial_{n}$ denoted by $Z_{n}(\mathcal{M})$, whose elements are called $n$-cycles. Meanwhile, the image of $\partial_{(n+1)}$ denoted by $B_{n}(\mathcal{M})$, consists of $n$-boundaries. The group given by $H_{n}(\mathcal{M})=Z_{n}(\mathcal{M})/B_{n}(\mathcal{M})$ is the $n$-homology group, which intuitively captures topological failures of dimension $n-1$ in $\mathcal{M}$. For instance, $Z_{1}(\mathcal{M})$ will consist of all one-dimensional objects in $\mathcal{M}$ that have no boundary, while $B_{1}(\mathcal{M})$ will consist of all boundaries of two-dimensional objects. Since the latter necessarily have no boundary themselves, $H_{n}(\mathcal{M})$ is capturing one-dimensional objects without boundaries that are not boundaries of any two-dimensional piece, identifying a failure in having such a piece, a ``missing'' point.

With $d$, referred to as the coboundary operator, we can utilize its duality with $\partial$ to show that $dd=0$, obtaining the cochain complex
\begin{equation}
    \begin{tikzcd}[row sep=tiny]
0 \arrow{r} & \mathcal{C}^{0}(\mathcal{M}) \arrow{r}{d_{0}} & \mathcal{C}^{1}(\mathcal{M}) \arrow{r}{d_{1}} & \mathcal{C}^{2}(\mathcal{M}) \arrow{r}{d_{2}} & \dots
\end{tikzcd}
\end{equation}
Where $\mathcal{C}^{n}$ are the duals of $\mathcal{C}_{n}$, which are infinitesimally generated by $n$-differential forms. The $n$-forms possess a beautiful mathematical structure which unfortunately will not be presented here. What will be important for us is that we can once again explore the topology of $\mathcal{M}$ with the $n$-forms, but this time not directly with pieces of $\mathcal{M}$ but rather with the functions that act upon these pieces. The kernel of $d_{n}$ is denoted by $Z^{n}(\mathcal{M})$, with elements called $n$-cocycles or closed $n$-forms, while the image of $d_{(n-1)}$ is denoted by $B^{n}(\mathcal{M})$, with elements called $n$-coboundaries or exact $n$-forms. The resulting algebraic structure $H^{n}(\mathcal{M})=Z^{n}(\mathcal{M})/B^{n}(\mathcal{M})$ is called the de Rham $n$-cohomology.

To conclude this brief presentation on elements of differential geometry, there are three theorems that will appear during this differential approach to contextuality and its immediate applications. The first is the generalized version of Stokes' theorem.

\begin{theorem}
Let $\omega$ be a smooth $(n-1)$-form with compact support on an oriented, $n$-dimensional manifold-with-boundary $S$, where $\partial S$ is given the induced orientation. Then
    \begin{equation}
    \int_{\partial S}\omega=\bra{\omega}\ket{\partial S}=\bra{d\omega}\ket{S}=\int_{S}d\omega.
\end{equation}
\end{theorem}

This theorem is the direct verification that the region $S$ can be seen as analogous to a region in $\mathbb{R}^n$, with the $n$-form $d\omega$ being in some way extended throughout the entire region.

Another important theorem is the Ambrose-Singer theorem, which relates the holonomy of a principal bundle to the curvature in the region where the holonomy is found. It identifies that the holonomy is the expression of curvature, and that curvature generates holonomy.

\begin{theorem}
     In a smooth manifold $M$ with a principal bundle $P$ over $M$ and a connection 1-form $\omega$, the holonomy algebra at a point $p\in P$ is generated by the curvature form $F$ derived from $\omega$ and evaluated along loops based at $p$. Specifically, the holonomy algebra is determined by the curvature form $F$ and its covariant derivatives evaluated on all possible pairs of horizontal vector fields at $p$.
\end{theorem}

The relationship between holonomy and curvature allows us to identify geometry through the phases of transport in loops. The curvature associated with the 1-form $\omega$ is nonzero at a point in $M$ if and only if there exists a nontrivial closed curve passing through that point whose holonomy phase along it is nontrivial.

Lastly, we have the Poincaré Lemma, which states that in a contractible manifold isomorphic to a region of $\mathbb{R}^n$, all closed forms are exact.

\begin{theorem}
Let $M$ be a smooth, orientable manifold of dimension $n$ that is isomorphic to $\mathbb{R}^n$. Then for every closed differential form $\omega$ of degree $k$ on $M$ (i.e., $d\omega = 0$), there exists a differential form $\eta$ of degree $k-1$ on $M$ such that $d\eta = \omega$.
\end{theorem}

\subsection{Discrete Differential Geometry}
\label{3.2}

As is customary in foundational studies, we initially confine ourselves to dealing with finite sets to elucidate non-classical behavior. This means that contextuality is initially addressed in finite structures, as are the usual versions of generalized contextuality. Furthermore, many GPTs, such as the classical one, feature edges that are obviously not smooth. Therefore, we need to take a step back and seek a way to incorporate contextuality into a formulation that is compatible with these conditions. This is the role of discrete differential geometry, which we will now briefly review.

An operational equivalence, as defined through a linear condition as above, encodes a discrete loop in its respective space. Contextuality is expressed in how the functions $\xi$, $\Gamma$, and $\mu$ deal with such loops. The noncontextual conditions can be thought of as discrete parallel transport of the probability functions that present no phase in a closed loop. Contextuality, as the violation of such conditions, is the discrete phase in each set. To formalize it, we will employ discrete differential geometry, or DDG \cite{Crane:2013:DGP,grady2010discrete}.

The formalism of DDG arises from the need for discrete methods to describe approximately smooth manifolds, as in computer graphics and geometry processing. While it's always possible to triangulate a smooth manifold, DDG, without imposing the usual differential structure, enables the study of more general topological manifolds known as piecewise linear manifolds.

We begin with a piecewise linear manifold $\mathcal{M}=\bigcup_{n}\mathcal{C}_{n}$, composed of sets of $n$-simplices $\mathcal{C}_{n}$. An $n$-simplex is treated as an $n$-dimensional ``unit of space'', and the topology is derived from the topology of the simplicial complex. To be valid as an approximation, each simplex is regarded as the tangent space of a point in a hypothetical smooth manifold. For calculus operations on this simplicial complex, we can introduce the formalism of discrete differential forms, which can be loosely understood as a method to quantify the ``size'' of the simplices. Discrete differential forms are defined as linear duals of the simplices, denoted by $\mathcal{C}^{n}$ representing the set of $n$-forms. If $\omega\in\mathcal{C}^{n}$, then we have
\begin{align}
\begin{split}
  \omega:\mathcal{C}_{n} &\to R\\
  \ket{S} &\mapsto \left<\omega |S\right>=\int_{S}\omega.
 \end{split}
\end{align}

The first operator in DDG is the boundary $\partial:\mathcal{C}_{n}\to\mathcal{C}_{n-1}$, defined as usual by the orientation defined in the simplicial complex
\begin{equation}
    \partial \{a_{1}a_{2}...a_{n}\}=\{a_{2}...a_{n}\}-\{a_{1}a_{3}...a_{n}\}+...\pm\{a_{1}a_{2}...a_{n-1}\}.
\end{equation}
As an example, a tetrahedron $\{abcd\}$ has boundary
\begin{equation}
    \partial \{abcd\}=\{bcd\}-\{acd\}+\{abd\}-\{abc\}
\end{equation}
in an alternate way. The second one is the coboundary $d:\mathcal{C}^{n}\to\mathcal{C}^{n+1}$. It is defined as the unique linear map that satisfies the generalized Stokes theorem for DDG
\begin{equation}
    \int_{\partial S}\omega=\bra{\omega}\ket{\partial S}=\bra{d\omega}\ket{S}=\int_{S}d\omega,
\end{equation}
where the bracket notation will be used to denote the action of a $n$-form on a $m$-dimensional region, $n\geq m$, both for the discrete and continuum cases. See that the integral gives a $(n-m)$-form as expected, and $0$-forms are identified as scalar.

The homology of the manifold $\mathcal{M}$ follows the simplicial homology \cite{Hatcher:478079}, and explores the topological structure of $\mathcal{M}$ through its simplicial complex structure and the boundary operator. Once the boundary of a boundary is empty, we have $\partial\partial=0$, and we also obtain the chain complex
\begin{equation}
    \begin{tikzcd}[row sep=tiny]
0 & \arrow{l}{\partial_0} \mathcal{C}_{0}(\mathcal{M}) & \arrow{l}{\partial_1} \mathcal{C}_{1}(\mathcal{M}) & \arrow{l}{\partial_2} \mathcal{C}_{2}(\mathcal{M}) & \arrow{l}{\partial_3} \dots
\end{tikzcd}
\end{equation}
The kernel of $\partial_{n}$, denoted by $Z_{n}(\mathcal{M})$, will have as its elements the $n$-cycles, while the image of $\partial_{(n+1)}$, denoted by $B_{n}(\mathcal{M})$, will have as its elements the $n$-boundaries. The algebraic structure $H_{n}(\mathcal{M})=Z_{n}(\mathcal{M})/B_{n}(\mathcal{M})$ is the $n$-homology. It explores the shape of $\mathcal{M}$ by directly studying the ``quanta of space'', or equivalently, the tangent space of $\mathcal{M}$ that is locally isomorphic to itself. Non-trivial $n$-holonomy implies an $n$-dimensional topological failure in the simplicial complex.

On the other hand, cohomology deals with the dual of the simplices, the discrete differential forms, and the coboundary operator. By the property $dd=0$, which follows from the definition of the coboundary using the generalized Stokes theorem, we have the cochain complex
\begin{equation}
    \begin{tikzcd}[row sep=tiny]
0 \arrow{r} & \mathcal{C}^{0}(\mathcal{M}) \arrow{r}{d_{0}} & \mathcal{C}^{1}(\mathcal{M}) \arrow{r}{d_{1}} & \mathcal{C}^{2}(\mathcal{M}) \arrow{r}{d_{2}} & \dots
\end{tikzcd}
\end{equation}
The kernel of $d_{n}$ is denoted by $Z^{n}(\mathcal{M})$, and it comprises the $n$-cocycles, also known as closed $n$-forms. The image of $d_{(n-1)}$, denoted by $B^{n}(\mathcal{M})$, comprises the $n$-coboundaries, also called exact $n$-forms. The algebraic structure $H^{n}(\mathcal{M})=Z^{n}(\mathcal{M})/B^{n}(\mathcal{M})$ is the de Rham $n$-cohomology. It involves studying what we integrate on $\mathcal{M}$ and how it relates to the shape of $\mathcal{M}$.

\section{Differential Geometry of Contextuality}
\label{4}

Cohomology studies the failure of solutions to equations of the form $d\omega=\sigma$, which reside in the cotangent space. Generally, the equation $d\omega=\sigma$ seeks a global cause for the local relationship it describes. Our aim here is to identify that contextual behavior is the inability to obtain such a global solution.

To this end, we demonstrate that it is possible to represent it in its usual form in discrete differential geometry, with the possibility of using differential geometry as a generalization. Such a description allows us to have a more natural view of the phenomenon of contextuality in ontic representations, both because it is mathematically akin to the well-known mathematics in physical applications and because it allows for the direct interpretation of mathematical objects regarding what is happening with the physical model.

\subsection{Operational Equivalences and Contextuality}
\label{4.1}

Effects, states, and transformations live in a real vector space by construction, which is isomorphic to its own tangent space. They form convex subsets through the imposition of convex combinations. This naturally gives rise to piecewise linear manifolds embedded in real vector spaces. Vectors are oriented pieces, quantities, of the manifold, and therefore reside in $\mathcal{C}^1$.

The equations \ref{basestates}, \ref{baseeffects} and \ref{basetrasformations} are just saying that an operational equivalence is a closed discrete loop $\gamma$ in the tangent space, 
\begin{equation}
 \sum_{s}a_{s}^{(\alpha)}\ket{P_{s}}=\gamma^{(\alpha)},
\end{equation}
\begin{equation}
 \sum_{r}b_{r}^{(\beta)}\ket{E_{r}}=\gamma^{(\beta)},
\end{equation}
\begin{equation}
 \sum_{t}c_{t}^{(\tau)}\ket{T_{t}}=\gamma^{(\tau)}.
\end{equation}
Operational equivalences and closed loops generated by elements of each subset encode the same information. What was described as two processes generating the same final process here is depicted as two distinct paths departing and arriving at the same point. The rearrangement under linear conditions turns the paths into a loop.

The noncontextual conditions presented in equations \ref{NCstates}, \ref{NCeffects}, and \ref{NCtransformations} are defined by probabilistic functions $\xi$, $\Gamma$, and $\mu$, indexed by ontic variables $\lambda\in\Lambda$. They live in the cotangent space as differential forms, acting on vectors and giving us a probability. Rewriting equations \ref{NCstates}, \ref{NCeffects}, and \ref{NCtransformations}, we obtain
\begin{equation}
\label{NCstates2}
    \phi^{(\alpha)}=\sum_{s}a_{s}^{(\alpha)}\mu_{\lambda}(P_{s})=\bra{\mu_{\lambda}}\left(\sum_{s}a_{s}^{(\alpha)}\ket{P_{s}}\right)=0,
\end{equation}
\begin{equation}
\label{NCeffects2}
    \phi^{(\beta)}=\sum_{r}b_{r}^{(\beta)}\xi_{\lambda'}(E_{r})=\bra{\xi_{\lambda'}}\left(\sum_{r}b_{r}^{(\beta)}\ket{E_{r}}\right)=0,
\end{equation}
\begin{equation}
\label{NCtransformations2}
    \phi^{(\tau)}=\sum_{t}c_{t}^{(\tau)}\Gamma_{\lambda'\lambda}(T_{t})=\bra{\Gamma_{\lambda'\lambda}}\left(\sum_{t}c_{t}^{(\tau)}\ket{T_{t}}\right)=0,
\end{equation}
for all $\lambda$ and $\lambda'$. The noncontextual conditions become just the valuation of the $1$-forms given by the functions $\xi$, $\Gamma$ and $\mu$ in each space given by the ontic representation to preserve the flat behavior of the vector spaces involved. In other words, we can understand such functions as potential vector fields in our discrete space, and ask for the preservation of the convex combination in the sense that the phase $\phi =0$ when evaluated on a loop $\gamma$.

A comment on linearity in the forms. The map that defines the vectors is $E\mapsto\ket{E}$, but generally, a vector $\sum_{r}c_{r}\ket{E_{r}}$ is different from $\ket{\sum_{r}c_{r}E_{r}}$, since the latter can lie outside $\mathcal{E}$. It can also include negative elements, so even the linear operations in equations \ref{basestates}, \ref{baseeffects}, and \ref{basetrasformations} generally lie outside of $\mathcal{E}$. This is necessary to deal with noncontextuality: the objective is to classically complete the theory, thus embedding it into a classical one, where $\sum_{r}c_{r}\ket{E_{r}}=\ket{\sum_{r}c_{r}E_{r}}$.

An important part of this construction is the immediate application in non-discrete loops. For any loop $\gamma$, we can integrate the differential form representing the valuation function of the ontic representation, obtaining a phase $\phi$. As we will see, in quantum theory this phase is already a well-studied non-classical phenomenon.

\subsection{Schrödinger's View: Geometric Contextual Behavior}
\label{4.2}

Let's keep our model in a flat space, such as understanding it as a submodel of a classical theory. This assumption is the trivial extension of the convex set to all vector spaces without any topological failure. Without such failures, all loops are just boundaries, $\gamma=\partial S$, and noncontextuality conditions can be rewritten as
\begin{equation}
    \bra{\mu_{\lambda}}\ket{\partial S_{\alpha}}=0,
\end{equation}
\begin{equation}
\bra{\xi_{\lambda'}}\ket{\partial S_{\beta}}=0,
\end{equation}
\begin{equation}
\bra{\Gamma_{\lambda'\lambda}}\ket{\partial S_{\tau}}=0,
\end{equation}
in the language of differential forms. Here we can use Stokes theorem to define the coboundary operator and get
\begin{equation}
    \bra{d \mu_{\lambda}}\ket{S_{\alpha}}=0,
\end{equation}
\begin{equation}
\bra{d \xi_{\lambda'}}\ket{S_{\beta}}=0,
\end{equation}
\begin{equation}
\bra{d \Gamma_{\lambda'\lambda}}\ket{S_{\tau}}=0.
\end{equation}
Again, this is possible because we are in $\mathbb{R}^n$, with states, effects, and transformations represented by vectors in its tangent space, making sense of $S_{\alpha}$, $S_{\beta}$, and $S_{\tau}$ as regions in $\mathbb{R}^n$.

In these conditions, every closed differential form is exact: if $\bra{d \xi_{\lambda'}}\ket{S}=0$ for all regions $S$, then $d\xi_{\lambda}'=0$, which means it is closed and thus exact, $\xi_{\lambda'}=dc_{\lambda'}$ with $c_{\lambda'}$ a function. The failure of noncontextual conditions implies that $\xi_{\lambda'}=dc_{\lambda'}+\omega_{\lambda'}$, where $d\omega_{\lambda'}\neq 0$, and by
\begin{equation}
\begin{split}
    \bra{d \xi_{\lambda'}}\ket{S_{\beta}}&=\bra{ddc_{\lambda'}}\ket{S_{\beta}}+\bra{d\omega_{\lambda'}}\ket{S_{\beta}}\\
    &=\bra{d\omega_{\lambda'}}\ket{S_{\beta}}
\end{split}
\end{equation}
we see that $\omega_{\lambda'}$ is the connection that generates contextual behavior, and $F_{\lambda'}=d\omega_{\lambda'}$ as the curvature 2-form. The same holds for states and transformations.

\begin{theorem}
Noncontextuality for measurements (transformations; states) of an ontic representation is equivalent to a null contextual curvature $0=F_{\lambda'}=d\xi_{\lambda'}$ (respectively $0=F_{\lambda'\lambda}=d\Gamma_{\lambda'\lambda}$; $0=F_{\lambda}=d\mu_{\lambda}$) for all hidden variables that index it.
\end{theorem}

Geometrically, we can view each valuation and set of objects as defining a fiber bundle, with $\mathbb{R}$ regarded as a commutative group. As an $\mathbb{R}$-bundle, it is isomorphic to the trivial bundle $\mathbb{R}^{n}\times\mathbb{R}$, and with the restriction $\mathcal{E}\times [0,1]$ well defined (and analogously for $\mathcal{T}$ and $\mathcal{P}$). The curvature $F$ is in the lifting by the valuation function of the set of objects in the fiber bundle.

In this view there is no topological failure; it is a geometrical question. It is analogous to electromagnetism, with an electromagnetic tensor $F$ that can be written through holonomic loops \cite{Sarita2016,weatherall2016categories}. The geometrical view identifies contextuality with non-trivial holonomy of the contextual connection $\omega$.

\begin{example}
The Sheaf approach \cite{Abramsky_2011}, as well as the Kochen-Specker contextuality \cite{Kochen1975}, impose a classical structure on local events by forcing them to be elements of their Boolean completion. In the Sheaf approach, this is done by imposing the sheaf properties on the presheaf of events that encode the outcomes. Since events are treated deterministically, they are identified as vertices in a classical GPT. Once completed in a Boolean structure given by the classical GPT, any contextual behavior is expressed in the valuations. For example, it is the presheaf that defines the distributions that need to be studied to verify the model's contextuality. Therefore, in its construction, they use Schrödinger's view.
\end{example}

\subsection{Heisenberg's View: Topological Contextual Behavior}
\label{4.3}

Let's reject the use of curvature to explain contextuality. This means that we want a valuation to satisfy the properties of a classical probability distribution, satisfying Kolmogorov's axioms. Thus, $F=0$, and contextuality is not a correction in the valuation but lies in a different part of the model.

\begin{theorem}
If $F=0$, then contextuality of an ontic representation is equivalent to monodromy.
\end{theorem}

\begin{proof}
Contextuality implies a correction in the valuation, once that by construction the form $dc$ satisfies the noncontextuality conditions. Thus the valuation must be $dc+\omega$ with $\phi=\bra{\omega}\ket{\gamma}$. As $F=d\omega=0$, the $1$-form must be closed but not exact to show any non-trivial phase $\phi$, which only happens when the loop $\gamma$ is not the boundary $\partial S$ of a region $S$ when seen by the valuation. In other words, the valuation cannot be defined in $S$, implying that the ontic representation in $\mathbb{R}^{n}$ does not preserve the topology induced by the set of objects $\mathcal{E}$, $\mathcal{P}$ or $\mathcal{T}$. The form $\omega$ capts such topological failures through monodromy $\phi$, once we cannot access these topological failures.
\end{proof}

Without curvature, we still need to define a correction $\xi_{\lambda'}=dc_{\lambda'}+\omega_{\lambda'}$, with $F=d\omega_{\lambda'}=0$. But now closed forms cannot be exact, which means $\omega_{\lambda'}$ is a representation of a topological failure. Specifically, it represents a non-trivial element of the first cohomological group $[\omega_{\lambda'}]\in H^1$ defined on the set of objects. In the topological view, even with the fiber $\mathbb{R}$ and with the restriction $\mathcal{E}$ well defined (and analogously for $\mathcal{T}$ and $\mathcal{P}$), the fiber bundle is not trivial. The basis is not topologically trivial, and so is the fiber bundle. And this is what the valuation detects.

The topological view allows us to generalize results from the standard contextuality framework to the generalized one \cite{Sidiney_2021}.

\begin{theorem}
The $\mathbb{R}$-fiber bundle described by a model on an ontic representation is trivial in the topological view and so noncontextual if and only if any local section admits an extension to a global section.
\end{theorem}

This result follows from the equivalence of the extendability of local sections and triviality para o $\mathbb{R}$-fiber bundle. The ontic representation is noncontextual for a given valuation if and only if the fiber bundle presents no phase, which is equivalent to present extensions to global sections for any ontic variable $\lambda$ (or the pair $\lambda$ and $\lambda'$ for transformations), implying the fiber bundle being trivial.

\begin{example}
In quantum theory, Gleason's theorem relates the properties of the set of effects of a quantum system of dimension 3 or more with the quantum probabilistic valuation given by the Born rule. To do this, it imposes certain conditions.

The first is that the set of effects will be studied in its representation as rays of a Hilbert space. The second is that all states with probabilistic valuations over the set of effects are valid, thus defining the physical states. The third condition imposes the continuity of such valuation, thus relating the topology of the set to what the valuation is capturing. The last condition is the noncontextuality of the valuation, which is nothing more than the imposition that the valuation does not carry the contextuality of the effects, i.e., there is no curvature.

When making the ontic representation, the condition of no curvature is maintained, but the representation cannot capture all the details of quantum theory. We conclude that the Born rule in the standard representation of a quantum system arises from Heisenberg’s view, imposing all the contextuality in the topology of the set of effects. In such a topology, the contextuality is encoded, and it is by exploring its expressions that fundamentally quantum phenomena can be identified.
\end{example}

\subsection{The Nature of Contextuality: A Choice Between Topology and Geometry}
\label{4.4}

Contextuality is a property presented by an ontic representation. A theory will be contextual if and only if no ontic representation can describe it. A representation of a given set of objects, be it effects, states, or transformations, occurs in two ways: in the encoding of the processes; and in the encoding of the valuation function that acts on these processes to give us the probabilities. Taking all processes and levels of encoding into account, there are many different ways contextuality can be encoded in the representation. This shows its lack of empirical fundamentality and guides us to what is truly fundamental in the model. That's what we'll discuss here.

Contextual behavior can be encoded in each part of the process depending on the ontic representation, but that's not what we'll address here. What we aim to demonstrate here is the structure that the process/valuation levels possess, and the freedom of representation that it allows. We can codify what is going on with a diagram (here I will use $\mathcal{E}$, but the same can be said about $\mathcal{P}$ and $\mathcal{T}$), , where contextuality arises from the non-commutativity of the diagram
\begin{equation}
\label{diagram}
    \begin{tikzcd}
\mathcal{E} \arrow{rr}{p} \arrow[hookrightarrow]{dr}{i} & & \left[0,1\right] \\
& \mathcal{S} \arrow{ur}{\xi_{\lambda'}} &
\end{tikzcd}.
\end{equation}
The three elements, the system $\mathcal{E}$, the classical representation $\mathcal{S}$, and the target for valuation $[0,1]$, are all fixed, as is the map $p$. As contextuality is the failure of one of the maps to keep the data of the system, the inclusion $i$ or the valuation $\xi_{\lambda'}$ of the representation fail.

The first case, failure of inclusion, is the usual interpretation in terms of contextuality. This is the topological view, where we interpret that it's the inclusion of the system that causes the problem. The noncommutativity of the diagram is fundamental, and it cannot be understood in any other way than as a departure from our notion of reality, as in the participatory-realistic interpretations of quantum theory. A justification for its use is that a loop could not be immediately written as the boundary of a region since an inner region is not supposed to exist in the first place, at least not for every loop. Thus, the curvature could not be defined. To avoid such a problem, the extreme is to suppose $F=0$ everywhere, which one can interpret as the non-existence of a correction in the valuation changing the classical behavior. It is an intrinsic description, whose contextuality depends on the set of objects itself. Therefore, we have failures in reality itself, which defines the topological description as an participatory-realistic point of view.

The second case, the failure of valuation, imposes embedding. It's the inadequate valuation that causes problems. All the properties of the original system can be captured once one uses a modified valuation. It's not just hidden variables of the ontic representation; it also needs new rules to extract the probabilities. This is what happens in intrinsic-realistic interpretations of quantum theory, and this is the point of the geometrical view. The trivial fiber bundle is imposed, but curvature in the connection creates the phases by holonomy, following the Ambrose–Singer theorem. The geometrical view changes the valuation function by a generator of non-classicality. It can be thought of as curvature of the valuation on a set of classical objects. It's a modification of our classical laws by a hidden nature.

Both notions are equivalent, and one can argue that both causes can coexist. Choosing to what extent the topological and geometric causes generate non-classicality is just a matter of representing a deeper phenomenon: contextuality. While there are no ways to differentiate between different representations by verifying that they aren't faithful to some level of reality yet to be explored, for example, empirically verifying in which ontic representation the model works, it doesn't make much difference what is actually going on at the ontic level.

\section{Applications}
\label{5}

We are in a position to apply the understanding of the contextuality that the description by differential geometry provides us. The strategy in all of them is simple: to look for where the $1$-form that carries the contextual behavior is in the respective formalism and, respecting its constraints, explore where it is expressed.

Each application has a preference for one of the views we described earlier. It's important to keep in mind that this is a choice for building each application, but it's not the natural description in which contextual behavior is encoded. Such a description is indifferent to interpretations, and as we'll see, it presents itself differently for each of the interpretations of quantum theory.

\subsection{Contextual Fraction}
\label{5.1}

The contextual fraction \cite{Abramsky_2017,barbosa2019continuousvariable} was introduced in the Sheaf approach \cite{Abramsky_2011}. Usually, it is applied in response to the limitations of this approach, only exploring the contextuality of effects, known as measurement contextuality, and imposing outcome-determinism. To achieve this, we fix a state and do not apply any transformation. An ontic representation will have this form
\begin{equation}
    p(E)=\int_{\Lambda}\mu(\lambda)\xi(E|\lambda),
\end{equation}
with $\mu$ a measure on the set of ontic variables $\Lambda$. Due to the condition of outcome-determinism, the valuation function $\xi$ will assume only the values $0$ or $1$, and, for simplicity, we assume the finiteness of the sets.

With these constraints, the contextual fraction can be used in its usual form, and we can express the probability $p$ as a decomposition
\begin{equation}
    p(E)=(\text{NCF})p_{NC}(E)+(\text{CF})p_{SC}(E),
\end{equation}
with the noncontextual fraction (NCF) and the contextual fraction (CF), where $p_{NC}$ represents the probabilistic distribution of a noncontextual model, and $p_{SC}$ represents the probabilistic distribution of a strong contextual model, the part without any noncontextual contribution. For the definition of CF and NCF, we seek to maximize the quantity of NCF among the set of distributions in the ontic variables $\mu$ that can be used. Even with the maximization of NCF, the distributions $p_{NC}$ and $p_{SC}$ are not uniquely defined. The fractions also satisfy the sum property of probability
\begin{equation}
\begin{split}
    1&=\sum_{r}p(E_{r})\\
    &=(\text{NCF})\sum_{r}p_{NC}(E)+(\text{CF})\sum_{r}p_{SC}(E)\\
    &=\text{NCF}+\text{CF},
\end{split}
\end{equation}
exposing the meaning of CF and NCF being referred to as fractions.

Just like the Sheaf approach, contextual fraction follows Schrödinger's view, embedding the effects in a classical GPT. The outcome-determinism condition implies that the valuation function must have Boolean values, thus forcing the effects to be fixed on the vertices of the effect hypercube. To expose contextuality, we need to choose the probabilistic distribution over ontic variables that maximizes NCF.

The maximization is in terms of $\mu$, and it pertains to how the valuation of the ontic representation encodes the behavior of the effects, given by the probabilistic weight of $\mu$ when embedding the set of effects in the ontic representation. For the valuation function $\xi$, we have that the previous conditions fix the set $\Lambda$ and its relation with $\mathcal{E}$, thus also fixing $\xi$. This means that the $1$-forms are not subject to maximization, as they are intrinsic to the ontic representation previously fixed.

We can do the decomposition of the valuation $\xi=dc+\omega$
\begin{equation}
    p(E)=\int_{\Lambda}\mu(\lambda)\bra{dc_{\lambda}}\ket{E}+\int_{\Lambda}\mu(\lambda)\bra{\omega_{\lambda}}\ket{E},
\end{equation}
that also satisfies the sum property of probability
\begin{equation}
    1=\sum_{r}p(E_{r})=\sum_{r}\int_{\Lambda}\mu(\lambda)\bra{dc_{\lambda}}\ket{E_{r}}+\sum_{r}\int_{\Lambda}\mu(\lambda)\bra{\omega_{\lambda}}\ket{E_{r}}.
\end{equation}

Let's assume that the maximization of NCF has already been done. Identifying the contextual and noncontextual parts of the previous equation as those containing $\omega$ and $dc$, respectively, we get the fractions
\begin{equation}
    \text{NCF}=\sum_{r}\int_{\Lambda}\mu(\lambda)\bra{dc_{\lambda}}\ket{E_{r}},
\end{equation}
\begin{equation}
    \text{CF}=\sum_{r}\int_{\Lambda}\mu(\lambda)\bra{\omega_{\lambda}}\ket{E_{r}},
\end{equation}
and the probabilistic distributions
\begin{equation}
    p_{NC}(E)=\frac{\int_{\Lambda}\mu(\lambda)\bra{dc_{\lambda}}\ket{E}}{\sum_{r}\int_{\Lambda}\mu(\lambda)\bra{dc_{\lambda}}\ket{E_{r}}},
\end{equation}
\begin{equation}
    p_{SC}(E)=\frac{\int_{\Lambda}\mu(\lambda)\bra{\omega_{\lambda}}\ket{E}}{\sum_{r}\int_{\Lambda}\mu(\lambda)\bra{\omega_{\lambda}}\ket{E_{r}}},
\end{equation}
which are not unique since they depend on the choice of $\mu$ in the set that maximizes NCF.

In particular, a property of the contextual fraction is that the maximal violation of a generalized Bell inequality of the model is given by CF \cite{Abramsky_2017}. Since $\text{CF}=\sum_{r}\int_{\Lambda}\mu(\lambda)\bra{\omega_{\lambda}}\ket{E_{r}}$, we have the explicit dependence on the contextual correction of the valuation.

\subsection{Interference in Quantum Measure Theory and Quantum Theory}
\label{5.2}

Interference is a natural property of quantum theory when described by wave functions. Richard Feynman said that interference ``contains the only mystery'' of quantum theory \cite{FeynmanVol3}. Rafael Sorkin \cite{1994SORKIN,sorkin1995quantum} introduced a generalized notion of interference as a correction to the standard measure theory based on Kolmogorov axioms, by modifying the disjoint rule to
\begin{equation}
    p(A\sqcup B)=p(A)+p(B)+I_{2}(A,B)
\end{equation}
for two disjoint sets,
\begin{equation}
\begin{split}
    p(A\sqcup B\sqcup C)=&p(A)+p(B)+p(C)\\&-I_{2}(A,B)-I_{2}(B,C)-I_{2}(A,C)\\&+I_{3}(A,B,C)
\end{split}
\end{equation}
for three disjoint sets, and so on, following analogously to the inclusion-exclusion principle in combinatorics.

With valuation functions that allow interference terms $I_n$, we have a generalization of the usual measure theory, which is called quantum measure theory. The set on which the probability measure acts is taken as a measurable set, and since we are dealing with the finite case, it has a natural Boolean structure. This implies that we are embedding the processes into a classical GPT and imposing that any correction must be expressed in the valuation function, specifically in the interference terms $I_n$. Therefore, quantum measure theory follows Schrödinger's view.

Here our interest will be in how such formalism can be used to capture contextuality, as already explored in Refs. \cite{2006Craig,2008Dowker,2014Dowker}.

\begin{example}
Quantum theory presents only the correction $I_{2}$ fundamentally non-trivial, which can be seen clearly in the path integral approach. The form of $I_{2}$ follows directly from the calculation of quantum probabilities from the Born rule. For sharp effects $e$ and $e'$, which we assume are not in the same subspace to be disjoint, and a pure state $\rho=\ket{\rho}\bra{\rho}$, we can write
\begin{equation}
    \begin{split}
        |\bra{e\vee e'}\ket{\rho}|^2&=\bra{e\vee e'}\ket{\rho}+\bra{\rho}\ket{e\vee e'}\\
        &=|\bra{e}\ket{\rho}|^2+|\bra{e'}\ket{\rho}|^2+I_{2}^{\rho}(e,e')
    \end{split}
\end{equation}
with 
\begin{equation}
    I_{2}^{\rho}(e,e')=\bra{e}\ket{\rho}\bra{\rho}\ket{e'}+\bra{e'}\ket{\rho}\bra{\rho}\ket{e}
\end{equation}
the symmetric interference function. 

The relationship with the interference of a double-slit experiment follows from treating each effect as a distinct path from the slit, which is explored by a state $\rho$. In this case, we can identify the isolated terms $|\bra{e}\ket{\rho}|^2$ and $|\bra{e'}\ket{\rho}|^2$ as closed paths that loop back on themselves. As for the term $I_{2}^{\rho}(e,e')$, it comprises two terms accounting for the two directions in which we can loop through the two paths. Interpreting them as ``passing through both paths simultaneously'' is what generates the notion of interference, as well as the non-classicality that such a term may carry.

For sharp effects, only $I_{2}$ is non-trivial, which follows from Specker's Principle \cite{speckervideo,cabello2012speckers}. For unsharp ones, high-order interference appears, but they are non-fundamental once they can be rewritten from the intersections ($\wedge$) and $I_{2}$.
\end{example}

Interference terms are necessary but not sufficient for contextualilty in a model. For instance, classical systems can exhibit interference, imposing a correction term on distributions. This term is given by classical correlations that are generated when marginalizing a GPT that captures all classical details into a smaller GPT through a marginalization process that acts as coarse graining. The difference is that the terms $I_n$ need not be non-negative, which, for the case of outcome-determinism and restricted to effects, is a sign of contextualilty \cite{Abramsky_2011,Spekkens2008}.

For the ontic representation given by quantum measure theory, the set of objects in the ontic representation is such that any correction to Kolmogorov's axioms follows from the connection $\omega$. This follows from the property of $dc$ satisfying the axioms in the ontic representation when subjected to the condition of outcome-determinism, as it arises from a probabilistic distribution over a set of effects of a Boolean GPT. Any part of $\omega$ that ends up being classical follows from the non-refinement of the ontic representation, as discussed in the previous paragraph.

Aware of this limitation, we can see how $\omega$ appears as the generator of the interference term. Suppose the effects $E$, $E'$, and $E\vee E'$ are not in a Boolean sub-algebra. Since $dc_{\lambda'}$ satisfies Kolmogorov axioms, we have that
\begin{equation}
\begin{split}
    p(E\vee E')=&\int d\mu \bra{\xi}\ket{E\vee E'}\\
    =&\int d\mu \left(\bra{dc}\ket{E\vee E'}+\bra{\omega}\ket{E\vee E'}\right)\\
    =&\int d\mu \left(\bra{dc}\ket{E}+\bra{dc}\ket{E'}+\bra{\omega}\ket{E\vee E'}\right)\\
    =&\int d\mu\bra{\xi}\ket{E}+\int d\mu\bra{\xi}\ket{E'}\\
    &+\int d\mu\left(\bra{\omega}\ket{E\vee E'}-\bra{\omega}\ket{E}-\bra{\omega}\ket{E'}\right)\\
    =&p(E)+p(E')+\\
    &\int d\mu\left(\bra{\omega}\ket{E\vee E'}-\bra{\omega}\ket{E}-\bra{\omega}\ket{E'}\right).
\end{split}
\end{equation}
So, for disjoint effects,
\begin{equation}
    I_{2}(E,E')=\int d\mu \left(\bra{\omega}\ket{E\vee E'}-\bra{\omega}\ket{E}-\bra{\omega}\ket{E'}\right)
\end{equation}
showing that the failure of $\omega$ to satisfy the Kolmogorov disjoint axiom is the cause of interference. Note that what is measured in the geometrical view is the failure of the parallelogram law of the valuation function captured by the distribution $\mu$ that represents the relationship between the state and the ontic representation. The analogous construction can be made with high-order interference.

As shown in \cite{anastopoulos2002quantum}, in quantum theory, the decoherence functional \cite{PhysRevD.46.1580} is determined by geometric phases. However, only its real part possesses ``reality'' due to being Hermitian, a property that arises from strong positivity \cite{https://doi.org/10.25560/70797}. This real part induces interference, which is the primary focus of quantum measure theory \cite{1994SORKIN,sorkin1995quantum,surya2008quantum,Gudder2009}. Therefore, the connection between interference and geometric phase is profound in quantum theory, and it can be utilized to detect non-classicality \cite{Asadian2015}.

\begin{example}
Noncommutativity in sharp quantum theory is where contextual behavior on effects is hidden. This follows from the capacity of using unitary transformations in relation to a fixed effect to define any other effect. As incompatibility of sharp effects is necessary for contextuality, and it is equivalent to noncommutativity \cite{Heinosaari_2008,Heinosaari_2010}, the non-exact part $\omega$ of the valuation depends on the non-trivial commutator. 

For two noncommutative unitaries, $U$ and $U'$, the structure constant depends on $\hbar$, and defines a loop $U^{-1}U'^{-1}UU'$. Noncommutativity implies a geometric phase that can generate an interference correction, thus given by the non-exact term $\omega$. The limit $\hbar\to 0$ cancels the non-classical behavior, which means we can write $\omega=\hbar\Tilde{\omega}$, to explicitly denote its dependence on $\hbar$,
\begin{equation}
    \xi_{\lambda'}=dc_{\lambda'}+\hbar\Tilde{\omega}_{\lambda'}.
\end{equation}

This also holds for states and transformations, which can exhibit non-classicality, as exemplified by the Bargmann invariants \cite{Bargmann} and the Aharonov–Bohm effect \cite{2010Popescu}. For this, one can use unitary transformations on them, analogous to what has already been done in addressing the relationship between contextuality and the Wigner function representation of quantum theory \cite{kocia2017,kocia2017again,kocia2018}.
\end{example}

\subsection{Signed Measures and Embedding of GPTs}
\label{5.3}

The relationship between contextuality and negativity has already been explored in Refs. \cite{Abramsky_2011, Spekkens2008}, as well as the relationship between embedding and contextuality in Refs. \cite{shahandeh2021, Schmid2021}. The unification of contextuality, embedding, and negativity was achieved in Ref. \cite{schmid2020structure}, within the formalism of GPTs. Our goal here is to discuss how the term $\omega$ and the two views on its nature allow us to understand such relationships, given that a classical ontic representation is defined by the commutation of diagram \ref{diagram}, and each view of noncommutativity explains the origin of these non-classical notions.

As shown in Ref. \cite{Spekkens2008}, the necessity of signed measures for all ontic representations is equivalent to the violation of the noncontextual conditions, thus implying the existence of non-trivial phases in the valuation maps. It is in the correction term, which arises from the curvature in Schrödinger's view, that codifies the negative part of the valuation function. Heisenberg's view explains why we cannot access negative probabilities. When translated into Heisenberg's view, the negative values of the valuation become topological failures; therefore, they are not within the set of physical processes. In the geometrical view, a classical model cannot have non-trivial curvature to correct the valuation, and any theory that exhibits such curvature for all ontic representation cannot be represented by a classical model. Since the necessity of negative values implies the existence of curvature that does not preserve the valuation function, we conclude that it characterizes the contextuality of the theory.

The embedding of the model into a classical GPT is equivalent to an ontic representation. In GPTs, such an embedding is a set of linear maps, each mapping to a set of objects that form the processes, such that the valuation is preserved. One way to determine if a theory is contextual is to verify that for every embedding in any classical GPT, a valuation map is not preserved. Another way is to show that if the map preserves the valuation functions, then for any embedding, at least one set of objects will be larger than the respective set of objects in the classical model. In the topological view, a classical model exhibits no topological failure, thus monodromy is impossible, and a theory with monodromy for all ontic representations cannot be represented by a classical model. The impossibility of embedding while preserving the valuation function implies the expression of these failures by imposing the valuation on them, which also characterizes contextual behavior.

All these non-classical notions are just different ways to explain to our classical eyes what contextuality is. They are just different representations of this phenomenon.

\subsection{Generalizing Vorob'ev Theorem}
\label{5.4}

For measurement contextuality satisfying outcome-determinism and no-disturbance, the measurements can be represented by a hypergraph of compatibility, where measurements are vertices and the contexts with mutually compatible measurements are the hyperedges connecting them\footnote{See Ref. \cite{Sidiney_2021} and references therein for the construction of these scenarios.}. Vorob'ev theorem \cite{Vorobyev_1962} imposes conditions for the hypergraph of compatibility to only describe noncontextual models. In the language of the hypergraph of compatibility, it takes the following form:

\begin{theorem}
 A hypergraph of compatibility $\left(\mathbf{M},\mathcal{C}\right)$ has only noncontextual behavior if and only if it is acyclic, which means it can be reduced to the empty set by the operations known as Graham's reduction:
\begin{itemize}
\item if $m\in C$ belongs to only one hyperedge, then delete $m$;
\item if $C \subsetneq C'$, with $C, C' \in\mathcal{C}$, then delete $C$.
\end{itemize}
\end{theorem}

We can rethink what these operations mean for the structure of the measurement's effects. Remembering that we are considering each context with a finite set of outcomes. Let us also assume, without loss of generality, that the set of outcomes is the same for the measurements. A context is a $\sigma$-algebra of local events, with the elements of the $\sigma$-algebra having the structure of a set of effects of a classical GPT, thus a Boolean structure.

The separation into contexts and outcomes turns the set of effects in a fiber bundle with a base set given by the compatibility hypergraph and the fibers as a finite $\sigma$-algebra. Vorob'ev theorem being restricted to the compatibility hypergraph is equivalent to identifying that the Boolean structure of the fibers in the bundle does not interfere with a necessary and sufficient condition for always having noncontextuality. Even as a property projected onto the base set, it is still a property of the total set, the set of effects.

The projection maps the $\sigma$-algebra of outcomes of a context to the $\sigma$-algebra given by the context and its subcontexts. The measurements are nothing more than the atoms of this resulting structure. With this fact in mind, it is easy to rewrite Graham's reduction:
\begin{itemize}
\item if an atom is just in one $\sigma$-algebra, we can ignore it.
\item if a sub-$\sigma$-algebra is proper, we can simplify it as a trivial $\sigma$-algebra, consequently coarse-graining the $\sigma$-algebra which contains it.
\end{itemize}
We can initiate this version of Graham's reduction without performing the projection, but by considering the outcomes of each measurement as the atoms. The projection would follow from the reduction itself, specifically from the second item, where the proper sub-$\sigma$-algebra is erased.

Measurement contextuality follows from a loop $\gamma$ in the set of effects. As any loop in a Boolean structure has a trivial contextual connection when viewed as a classical GPT, $\omega=0$, only loops defined through different Boolean structures can show any contextuality. Vorob'ev result identifies the fact that without such loops, no contextual behavior appears. For the argument to work, we need to make the dependence on contexts explicit even in the set of effects, placing us in Heisenberg's view, therefore $F=d\omega =0$.

\begin{theorem}[Generalized Vorob'ev]
A ontic representation is noncontextual if its first de Rham cohomological group is trivial and $F=0$.
\end{theorem}

\begin{proof}
Once $F = 0$, we need to use the topological view, and contextuality will appear as topological failures that cause the monodromy of the probabilistic valuation. If the first cohomological group is trivial, then there are no differential forms that capture these failures, thus $\omega = 0$. Therefore, the ontic representation satisfies the noncontextuality condition.
\end{proof}

To identify the unavoidable contextual behavior of a model, we must exhaust all possible ontic representations. In other words, all possible valuation $1$-forms indexed by the ontic variables must present a cohomological obstruction. For the case of outcome-determinism, we have a fixed ontic representation, as we discussed earlier, which allows us to apply the previous theorem and obtain the noncontextuality of the model even before knowing the valuation functions. Acyclicity is a special case in which, still in the compatibility hypergraph, we can identify that there is no possibility of there being a loop that allows the expression of contextual behavior.

An important point is that $H^{1}$ is defined by the effects, not measurements. Thus, the intuition that the first homology groups of the compatible hypergraph must be non-trivial to show contextuality is false, even though the topological view holds that $\phi=\bra{\xi_{\lambda'}}\ket{\gamma}\neq 0$ is a failure also detected by $H_{1}$. For discussion and counterexamples, see Ref. \cite{Sidiney_2021}. A study on the topological expression of Vorob'ev's theorem will be conducted in a future work.

\subsection{Disturbance and Transition Functions}
\label{5.5}

Disturbing models for measurement contextuality are not standard in almost all frameworks. A well-known exception is the Contextuality by Default approach \cite{Dzhafarov2015ContextualitybyDefaultAB}. Disturbance is necessary for describing experimental applications. One way to indirectly address them is to modify the scenario \cite{Amaral_2019}, which seems to reflect William James's point of view on contradiction: when encountering a contradiction, make a distinction. The idea is to make the contradiction that disturbance represents explicit in the scenario by adding new maximal contexts representing the disturbing intersections. Due to its nature of being a discordance in the intersections between contexts, it is natural to relate disturbance to transition maps in a suitable approach where the latter are defined, such as in the bundle approach \cite{Sidiney_2021}.

A measurement is given by its effects, seen as possible events, not necessarily deterministic. In the deterministic case, the set of effects has a natural $\sigma$-algebra structure. Seen in this way, no-disturbance is exactly the triviality of the transition maps on intersections of $\sigma$-algebras. The set of measurements covers the set of effects and can be seen as charts of an atlas. Each deterministic measurement is a classical GPT in itself, or in other words, there exists an embedding of it as an entire classical GPT.

The geometrical view, where the contextual phases follow from holonomy, provides a direct approach to deal with disturbance, given that holonomy can be encoded in an element of a commutative group, the group of automorphisms of $\mathbb{R}$. For the $1$-form $\omega$, one can express such a phase as stated in \cite{Holonomyvideo}
\begin{equation}
    Hol(\partial S)=exp\left(\bra{\omega}\ket{\partial S}\right)
\end{equation}
where we identify $\bra{\omega}\ket{\partial S}$ as an element of the Lie algebra of the Lie group of transformations of $\mathbb{R}$. For discrete cases, one can embed the group in the Lie group of transformations of $\mathbb{R}$ and the same argument follows. As a commutative group, one only needs to track each chart, here $\sigma$-algebra,
\begin{equation}
    Hol(\partial S)=exp\left[\sum_{r}\bra{\omega}\left(b_{r}\ket{E_{r}}\right)\right].
\end{equation}

Let's consider the disturbing case now. Due to the presence of disturbance, it is not possible to have a global ontic representation, as it erases disturbance since classical GPT is non-disturbing. It can only be done locally for each chart. Similar to in differential geometry, the non-triviality between intersections defines the transition maps $t_{r,r'}$, with the indices $r$ and $r'$ denoting the charts, which lead from one chart to another. The holonomy transformation can then be calculated, and it will be given by
\begin{equation}
    Hol(\partial S)=\prod_{r}exp\left[\bra{\omega}\left(b_{r}\ket{E_{r}}\right)\right]\prod_{r}t_{r,r-1},
\end{equation}
where commutativity was used to rearrange and combine the transition maps. They can be put in the Lie algebra form, $t_{r,r-1}=exp(\eta_{r,r-1})$, to rewrite the holonomy term as
\begin{equation}
    \sum_{r}\bra{\omega}\left(b_{r}\ket{E_{r}}\right)+\sum_{r}\eta_{r,r-1}.
\end{equation}

What the transition map is doing in the geometrical view is taking one classical GPT into another, carrying the effects on themselves, but changing the valuation. This means that just as in contextual behavior, a correction term that is sensitive to chart changes should be included in the valuation. For this adjustment, we can define a $1$-form $\ket{\eta}$ that satisfies $\eta_{r,r-1}=\bra{\eta}\left(b_{r}\ket{E_{r}}\right)$, and rewrite the valuation function as
\begin{equation}
    \xi=dc+\omega+\eta.
\end{equation}
The first term values the global contribution, the second term values dependencies on parallel paths, while the third term values changes in an effect when transitioning from one chart to another.

The disturbance form $\eta$ is on the same footing as contextuality, which is not surprising, as charts are nothing more than contexts. In a certain sense, we have that it is the same phenomenon, the dependence on contexts of the valuation. The difference is that $\omega$ deals with paths, while $\eta$ is pointwise on the effects. For this reason, just as the $1$-form $\omega$ depends on the ontic representation to be defined, the $1$-form $\eta$ depends on the atlas to be defined.

The explicit construction of examples of this formalism to deal with disturbance as transition maps will be done in a subsequent work. However, we can already identify its relationship with extended contextuality \cite{Amaral_2019}. Adding contexts and duplicating measurements that show disturbance is nothing more than turning the form $\eta$ into an effective form $\omega$ of an ontic representation without disturbance by making effects on the disturbed intersection into paths. Thus, the disturbance becomes contextuality. It can explain the fact that disturbance consumes contextual behavior, as already noted in the nonlocality literature \cite{Abramsky_2014,Blasiak_2021}.

In the same way as in contextual fraction, we can think of an analogous quantifier for disturbance. Let's suppose here that we already have the valuation $\xi$ defined in a fixed ontic representation, just as in the case of contextual fraction, which implies that we already have $c$, $\omega$, and $\eta$. The disturbance fraction (DF) will be given by
\begin{equation}    \text{DF}=\sum_{r}\int_{\Lambda}d\mu(\lambda)\bra{\eta_{\lambda}}\ket{E_{r}}
\end{equation}
through reasoning analogous to that of contextual fraction. Similarly, an induced maximally disturbing model is given by
\begin{equation}
    p_{D}(E)=\frac{\int_{\Lambda}d\mu(\lambda)\bra{\eta_{\lambda}}\ket{E}}{\sum_{r}\int_{\Lambda}d\mu(\lambda)\bra{\eta_{\lambda}}\ket{E_{r}}},
\end{equation}
with $p=(\text{NCF})p_{NC}+(\text{CF})p_{SC}+(\text{DF})p_{D}$, where $p_{D}(E)$ concentrates all the disturbance. The conditions for the explicit construction of this disturbance fraction, as well as its relation to other proposals in the literature \cite{valle2023corrected}, will be addressed in a subsequent work.

\subsection{Contextuality in Interpretations of Quantum Theory}
\label{5.6}

An interpretation of quantum theory does not aim solely to create models for isolated processes, but rather to create a consistent framework for the entire quantum theory. The goal is to present an ontic structure and probabilistic valuation functions capable of explaining quantum phenomena, even if some properties understood as classical need to be violated.

At the core of all interpretations lies the measurement problem, the fact that a measurement does not follow from the standard evolution of the theory. This problem arises from the incompatibility of measurements and the impossibilities they generate. Contextuality is at the root of the measurement problem; indeed, it is the phenomenon that prevents a classical ontology, and it is the phenomenon that interpretations indirectly deal with.

The number of interpretations grows exponentially, and it is beyond the scope of this work to address them individually. We will use the classification by Cabello \cite{Cabello_2017}, exploring some specific examples a bit more deeply. Our goal is to try to identify how the interpretations exhibit the correction of contextuality $\omega$. Their incompatibilities stem from choosing different ontic representations, or even where contextuality appears, following one of the views presented in section \ref{4}.

In \cite{Cabello_2017}, we have two types of interpretations, denoted as Type-I and Type-II. The Type-I interpretations presume an intrinsic realism, that is, they seek an ontological foundation. Within Type-I, there is a further distinction concerning the ontological value of the quantum state. The $\psi$-Ontic interpretations assign ontological value to the states. We will explore two well-known examples.

\begin{example}
    Bohmian mechanics \cite{Bohm1,Bohm2,BohmStanford} generates an interesting interpretation in which we add a purely nonlocal term to classical deterministic dynamics. The quantum state is taken as ontology and broken down into its local classical and purely nonlocal quantum parts, which ``pilot'' classical objects.

    The embedding of the theory into a classical framework is exact, with the correction given by the influence that the global part of the ontology exerts on the local classical part. This correction alters the valuation that an agent has access to, which is the restriction of classical determinism with dependence on global properties.

    The encoding in classical ontology shows that Bohmian mechanics utilizes Schrödinger's view. Describing an extra term to correct classical behavior, and that this term is nonlocal, embodies the correction that quantum contextual behavior imposes.
\end{example}

\begin{example}
    The Many Worlds Interpretation \cite{MW,MWStanford} takes the quantum state as the ontological object. Reality would be entirely described by a quantum state that evolves unitarily. Each state is possible as a classical reality, each possibility is a reality in some world.

    We immerse quantum theory in a classical theory, since we give reality to all states in a multiverse. The difficulty lies in retrieving the quantum valuation rule once the worlds are taken as real and can be treated as classical. This imposes restrictions on the classical distribution that introduce correlations between worlds, thus modifying the valuation function to fit the Born rule.

    The Many Worlds Interpretation is expressed through Schrödinger's view, which is natural since this interpretation follows from assuming the absolute reality of Schrödinger's formalism.
\end{example}

In Type-I interpretations, those that are not $\psi$-Ontic are $\psi$-Epistemic. Such interpretations remove the ontological value of the quantum state, proposing an ontology that completes this state, resulting in an epistemic value for it. Let's explore an example for this case.

\begin{example}
    Consistent Histories Interpretation \cite{CH,CHStanford} treats events in a classical manner. Quantum theory would be given through a stochastic process over events, forming histories. It is in the valuation that we have the expression of purely quantum behavior, with the existence of incompatible histories. To make a measurement, the agent needs to choose a set of compatible histories, called a framework, where the Kolmogorov axioms for classical probability hold.

    Thought of as an ontic representation of quantum theory, Consistent Histories Interpretation embeds the objects of processes in a classical structure. Classical valuation undergoes modification akin to that presented in quantum measure theory, altering stochastic behavior beyond the classical. A study of quantum theory in this interpretation seeks means of recovering quantum statistics through physical principles in valuation functions, since the events are already fixed.

    This is an interpretation that explicitly employs Schrödinger's view, with contextuality expressed in the geometric part given by the valuation of the ontic representation.
\end{example}

Type-II interpretations deal with a participatory realism, where the information extracted from a system is not intrinsic to it, but rather a result of its relationship with the observer. They also have a distinction into two types, but of a more epistemological nature. The Type-II interpretations that deal with knowledge are those that treat the quantum state as an object that encodes the observer's knowledge. Let's look at an example.

\begin{example}
    Relational interpretation \cite{Rovelli_1996,RQM} proposes that ontology lies not in states, but in a relational structure between agents and between processes. It assumes that state concerns the knowledge of an agent and is no longer fundamental but rather local. Consistency in this relational structure arises from an agent having restricted knowledge of the process when positioned as an observer within a system. For a third party, the system's evolution alongside the agent would still be the usual unitary evolution of quantum theory.

    There is no completion of the theory, and probabilities are treated as classical. What is lost is the absoluteness of the quantum state for each process as seen by each agent, yet with agreement enforced by an underlying fundamental structure. There are no local contradictions since the measurement is classical, but only when comparing such local measurements.

    This identifies this interpretation as utilizing Heisenberg's view, but concealing within the breakdown of the global into local the topological flaws that contextuality identifies.
\end{example}

The Type-II interpretations that deal with belief treat the quantum state not as knowledge but as an object that encodes the agent's expectation. Let's explore an example of this type of interpretation.

\begin{example}
    Quantum Bayesianism \cite{QB,QBStanford} takes a step further and presents an interpretation that avoids ontology. Agents' beliefs are optimized by imposing the Born rule valuation on the quantum state that represents those beliefs. This explicit non-realism positions the Born rule merely as a function that selects certain events, which only hold if they can be experimentally accessible. This highlights complete freedom regarding the structure of events.

    It is within events that the expression of contextuality occurs, as valuations are merely classical mechanisms that an agent employs. We are adopting Heisenberg's view, in a position where no ontology holds any meaning other than that of mere representation.
\end{example}

The examples in this section make explicit the relationship that the views presented in section \ref{4} have with the interpretations of a non-classical theory. Schrödinger's view is the intrinsic-realistic treatment of the model, with ontology being represented classically, and with all modifications occurring in the valuation of such ontic variables. It is in the geometry of this valuation that contextual behavior is found, usually as a correction term when forcing measurements outside the classical scope. Heisenberg's view is the participatory-realistic treatment, presupposing that ontology is not classical due to flaws in accessible propositions. These are ``holes'' in reality, and it can only be seen in pieces, by a covering of classical pieces. It is the identification of these ``holes'' as topological failures that allows contextuality to appear in valuation even if it is classical.

\section{Discussion}
\label{6}

In this work, we looked for an alternative description of contextuality in the generalized approach by constructing a representation of contextuality as a differential geometry problem.

The approach we used involved identifying the operational equivalences as discrete loops in the vector space where the objects are represented, and contextuality for them as the non-preservation of such discrete loops by the valuation functions. This identification allows the use of discrete differential geometry to deal with contextuality and naturally extends to non-discrete cases.

There are two different ways in this approach to understand contextual behavior, different views of contextuality. In the first case, the classical ontic representation is imposed, which implies the existence of a correction to the valuation function. In the second case, the classical probabilistic valuation is imposed, thus forbidding a correction without a fundamental cause, the non-triviality of the topology of the set of objects itself. Both notions are equivalent, and they are expressed in other concepts of non-classicality: contextual fraction, interference, signed measures, and non-embeddability. Rethinking contextuality of an ontic representation allows us to rethink and generalize the Vorob'ev theorem, and also gives us a natural way to deal with models that violate no-disturbance.

Contextuality is more than just topological \cite{Mansfield2020contextualityis}, even if it can be expressed as such if we assume Heisenberg's view. It is a higher-level phenomenon than that expressed in representations and interpretations. The choice of how to bring such a phenomenon to this level can be related to notions of realism. For example \cite{realismvideo}, if we consider Fixed Realism, where there is one model, and the real is what is true in it, then we impose that contextuality is no longer topological, but geometric as in Schrödinger's view. Even if we consider Covariant Realism, where we have equivalent models, and the real is how things change when one changes the model, still the global realism would impose a view with correction in the valuation, and not in the events. Now, if we consider Local Realism, where there are nonequivalent models, and the real is how to handle disagreement, then we can use this disagreement to bring contextuality to be encoded at the level of events, allowing the valuation not to be modified.

In the first two cases of realism, Fixed and Covariant, we have an intrinsic-realistic view, where what matters is the reality itself and how we see it. Examples of interpretations of quantum theory that are of this type are Many Worlds, Bohmian Mechanics, and Consistent Histories, with the latter being explicitly of the second case. The third case, Local Realism, is a participatory-realistic view, where the most important thing is a pragmatic description of what is observed. Contextuality is generally presented in this way, in its nonlocality version. The topological view explains that non-trivial cohomology is a signal of disagreement, in this case, contextual behavior. It is a matter of pragmatism, where if we accept the existence of new hidden features, we can recover the intrinsic-realistic \cite{rovelli2021}.

Some open questions for future exploration naturally follow. The exploration of the relationship between the curvature of the sets of states, effects and transformations that naturally appear in quantum theory, and the curvature in Schrödinger's view is an important step to make quantum contextuality explicit. More generally, a deep exploration of models and interpretations, with the explicit construction of their connection and curvature, or their topological failures. The importance of this exploration goes beyond the foundations of non-classical theories and their applications, as it would elucidate the main resource to be explored in emerging technologies. To make contact with such technologies, a next step would be to explore how the classical limit would be expressed with the approach presented here, and its relationship with mechanisms that erase the contextual connection. Furthermore, as disturbance is natural in experimental issues, the identification of transition maps as representations of disturbance is a way to bring the formalism of contextuality to such issues. And the identification also points to the exploration of higher holonomies and their possible relationship with the contextuality of higher-level processes. An explicit construction of the bundles, especially with the same language used in other areas of physics like field theory, would enable a greater understanding of the phenomenon of contextuality. In contact with areas of computation, an explicit construction of a topological characterization of the Vorob'ev theorem would also provide a more intuitive view of it.

If the contextuality presented by certain processes causes discomfort when thought of as something about the questions we can empirically ask, questions that we unjustifiably assume exist, we can change the point of view and faithfully represent the same phenomenon with all the questions we assume but with a correction in the answers. It's another lesson in humility that nature gives us, but also of our inventiveness. We can't ask the questions we want, we don't have the power to force it into an interrogation. We only receive answers that we are ready to receive, and only what it allows us to access. However, we can interpret the answers and represent the reality in which we live. Our confusion over quantum theory and its contextual behavior maybe stems more from our arrogance in forcing ourselves into the first case than our wisdom in adapting to the second case.

\begin{acknowledgments}
The author thanks the MathFoundQ – UNICAMP – Mathematical Foundations of Quantum Theory Group, in special Prof. Dr. Marcelo Terra Cunha.

This study was financed in part by the Coordenação de Aperfeiçoamento de Pessoal de Nível Superior - Brasil (CAPES) - Finance Code 001.
\end{acknowledgments}

\nocite{*}
\raggedright
\bibliography{aipsamp}

\end{document}